\documentclass[a4paper,onecolumn,11pt,accepted=2019-11-22]{quantumarticle}

\pdfoutput=1

\usepackage[usenames,dvipsnames]{xcolor}
\usepackage{amsfonts}
\usepackage{amsmath,amsthm,amssymb,dsfont}
\usepackage{enumerate}
\usepackage{graphicx}	
\usepackage[margin=3.5cm]{geometry}
\usepackage{url}
\usepackage{bbm}

\usepackage{hyperref}
\hypersetup{colorlinks=true,citecolor=blue,linkcolor=blue,filecolor=blue,urlcolor=blue,breaklinks=true}

\usepackage{nicefrac}

\theoremstyle{plain}
\newtheorem{theorem}{Theorem}[section]
\newtheorem{lemma}[theorem]{Lemma}

\newtheorem{corollary}[theorem]{Corollary}
\newtheorem{proposition}[theorem]{Proposition}
\newtheorem{conjecture}[theorem]{Conjecture}

\theoremstyle{definition}

\newtheorem{remark}[theorem]{Remark}

\newcommand*{\cB}{\mathcal{B}}

\newcommand*{\cE}{\mathcal{E}}

\newcommand*{\N}{\mathbb{N}}

\newcommand*{\St}[1]{\mathrm{S}(\mathbb{C}^{#1})}

\newcommand*{\eps}{\varepsilon}

\newcommand*{\id}{\mathrm{id}}

\newcommand*{\tr}{\mathrm{tr}}

\newcommand{\norm}[1]{\left\lVert#1\right\rVert}

\def\Tr{\operatorname{tr}}

\def\norm#1{\left|\!\left|{#1}\right|\!\right|}
\def\id{\mathsf{id}}

\def\id{\mathds{1}}

\def\>{\rangle}
\def\<{\langle}

\allowdisplaybreaks    

\begin{document}

\title{An information-theoretic treatment of quantum dichotomies}

\author{Francesco Buscemi}
\affiliation{Graduate School of Informatics, Nagoya University, Nagoya, Japan}
\orcid{0000-0001-9741-0628}
\author{David Sutter}
\affiliation{Institute for Theoretical Physics, ETH Zurich, Switzerland}
\orcid{0000-0001-9779-8888}
\author{Marco Tomamichel}
\affiliation{Centre for Quantum Software and Information and School of Computer Science, \newline University of Technology Sydney, Sydney}
\email{marco.tomamichel@uts.edu.au}
\orcid{0000-0001-5410-3329}

\date{}
\maketitle

\begin{abstract}
Given two pairs of quantum states, we want to decide if there exists a quantum channel that transforms one pair into the other. The theory of quantum statistical comparison and quantum relative majorization provides necessary and sufficient conditions for such a transformation to exist, but such conditions are typically difficult to check in practice. Here, by building upon work by Keiji Matsumoto, we relax the problem by allowing for small errors in one of the transformations. In this way, a simple sufficient condition can be formulated in terms of one-shot relative entropies of the two pairs. In the asymptotic setting where we consider sequences of state pairs, under some mild convergence conditions, this implies that the quantum relative entropy is the only relevant quantity deciding when a pairwise state transformation is possible. More precisely, if the relative entropy of the initial state pair is strictly larger compared to the relative entropy of the target state pair, then a transformation with exponentially vanishing error is possible. On the other hand, if the relative entropy of the target state is strictly larger, then any such transformation will have an error converging exponentially to one. 
As an immediate consequence, we show that the rate at which pairs of states can be transformed into each other is given by the ratio of their relative entropies. We discuss applications to the resource theories of athermality and coherence, where our results imply an exponential strong converse for general state interconversion. 
\end{abstract}

\section{Introduction}

Various pre- and partial orders have been the subject of extensive study both in mathematical statistics~\cite{inequalities34,blackwell1953,lecam1964,Torgersen1970,alberti1982stochasticity,torgersen1991comparison,DAHL199953} and in information theory~\cite{shannon1958note,korner-marton-75,cohen1998comparisons}. An example of paramount importance is that provided by the majorization preorder~\cite{marshall11}: a probability distribution $\vec{p}_1$ is said to majorize another distribution $\vec{p}_2$, in formula $\vec{p}_1\succeq \vec{p}_2$, whenever there exists a bistochastic\footnote{A transformation is said to be bistochastic if it transforms probability distributions to probability distributions, while keeping the uniform distribution fixed.} transformation $T$ such that $T\vec{p}_1=\vec{p}_2$. The majorization preorder is particularly relevant and useful because of a famous result by Hardy, Littlewood, and P\'olya, according to which the relation $\vec{p}_1\succeq\vec{p}_2$ can be expressed in terms of a finite set of inequalities~\cite{inequalities34} of the form $f_i(\vec{p}_1)\ge f_i(\vec{p}_2)$, intuitively capturing the idea that $\vec{p}_1$ is ``less uniform'' than $\vec{p}_2$. Such inequalities can be conveniently visualized by plotting the Lorenz curve of $\vec{p}_1$ versus that of $\vec{p}_2$~\cite{marshall11}.

As it involves the comparison of two probability distributions relative to a third one (i.e., the uniform distribution), the majorization preorder is naturally generalized by considering two pairs of probability distributions, that is, two \textit{dichotomies} $(\vec{p}_1,\vec{q}_1)$ and $(\vec{p}_2,\vec{q}_2)$, where now $\vec{q}_1$ and $\vec{q}_2$ are arbitrary distributions. One then writes $(\vec{p}_1,\vec{q}_1)\succeq(\vec{p}_2,\vec{q}_2)$ whenever there exists a stochastic transformation simultaneously mapping $\vec{p}_1$ to $\vec{p}_2$ and $\vec{q}_1$ to $\vec{q}_2$. As a consequence of Blackwell's equivalence theorem~\cite{blackwell1953}, also the more general case of dichotomies is completely characterized by a finite set of simple inequalities, which directly reduce to those of Hardy, Littlewood and Polya if $\vec{q}_1$ and $\vec{q}_2$ are both taken to be uniform. Also in this more general scenario, a \textit{relative} Lorenz curve can be associated to each dichotomy, and the preorder $\succeq$ visualized accordingly~\cite{Renes2016}.

In relation to quantum information sciences, while the preorder of majorization has found early applications in entanglement theory~\cite{Nielsen1999}, the notion of \textit{relative} majorization have started being employed only more recently, especially due to its applications in quantum thermodynamics~\cite{Horodecki2013, wehner13, buscemi-2015, Renes2016, Buscemi2017, Gour2018} and generalized resource theories~\cite{Chitambar2019}. In the quantum setting, the objects of comparison are \textit{quantum dichotomies}, namely, pairs of density matrices.

Let us consider two arbitrary finite-dimensional quantum dichotomies $(\rho_1,\sigma_1)$ and $(\rho_2, \sigma_2)$. In complete analogy with the classical case \textit{\`a la} Blackwell, the relative majorization preorder $\succeq$ can be extended to the quantum setting by writing $(\rho_1,\sigma_1)\succeq(\rho_2, \sigma_2)$ whenever there exists a completely positive trace-preserving map $\mathcal{E}$ such that $\mathcal{E}(\rho_1)=\rho_2$ and $\mathcal{E}(\sigma_1)=\sigma_2$ simultaneously. However, in the quantum case\footnote{With the notable exception of the qubit case~\cite{Alberti1980,chefles04} and pure states~\cite{chefles00,siddhu16}.}, there is no known simple set of inequalities, analogously to the comparison of two relative Lorenz curves, able to completely capture the relative majorization preorder~\cite{Reeb2011, Jencova2012, matsumoto2014example, Buscemi2017}: statistical conditions can be derived~\cite{Buscemi2012, matsumoto2010randomization, jenvcova2016comparison} but they typically involve an infinite number of inequalities, thus becoming much more cumbersome to verify.

In order to overcome such problems, in this paper we build upon an information-theoretic approach for the comparison of quantum dichotomies first considered by Matsumoto in~\cite{matsumoto10}. This involves the relaxation of the order $\succeq$ to allow for errors in the transformation, and the consideration of an asymptotic regime, in which an increasing number of identical copies of one dichotomy get transformed into something that resembles, up to an arbitrarily high level of accuracy, many copies of the other dichotomy. More precisely, while we allow for (small) errors in the transformation $\mathcal{E}(\rho_1)\approx\rho_2$, the condition $\mathcal{E}(\sigma_1)=\sigma_2$ must always be satisfied exactly, as if it were a sort of ``conservation rule.'' We compute that the optimal rate at which such transformation can happen is given by the ratio between the quantum relative entropies $D(\rho_1\|\sigma_1)$ and $D(\rho_2\|\sigma_2)$. Our result hence shows that quantum dichotomies, while enjoying a very rich structure, are asymptotically characterized by a single number, that is, their relative entropy.

The remainder of the paper is structured as follows. After introducing the relevant one-shot divergences and their properties in Section~\ref{sec:notation}, we establish our main technical results in Section~\ref{sec:main}. First, in Section~\ref{sec:main-exact} we derive sufficient conditions in terms of one-shot divergences for exact pairwise state transformations. In Section~\ref{sec:main-approximate} we relax this to allow for an error on one of the states, and again find sufficient conditions in terms of smoothed one-shot entropies. This then allows us to derive our main results in Section~\ref{sec:main-asymptotic}, Theorems~\ref{th:main} and~\ref{th:converse}, which together show that the relative entropy fully characterizes when pairwise transformations are possible asymptotically.
Section~\ref{sec:rates} then takes an information-theoretic approach to the problem by studying the maximal rate at which independent copies of states can be transformed into each other in Theorem~\ref{th:rates}.
We finally discuss applications of our results to the resource theories of athermality and coherence in Section~\ref{sec:app} and end with a conjecture characterizing the second-order asymptotic behavior of resource transformations in Section~\ref{sec:discussion}.

\emph{Note added.} Concurrent work~\cite{wang19,sagawa19} also derives some of our results in the setting of resource theories and thermodynamics. We explore these connections in Section~\ref{sec:app} for the convenience of the reader.

\section{Preliminaries}
\label{sec:notation}

\subsection{Notation} 
Let $\St{d}$ denote the set of \emph{quantum states} on a $d$-dimensional Hilbert space $\mathbb{C}^d$. A \emph{quantum channel} $\cE: \St{d} \to \St{d'}$ is a linear map that is \emph{completely positive} and \emph{trace-preserving} (CPTP). We denote by $\leq$ the L\"owner partial order, i.e., for two Hermitian matrices $X$ and $Y$ the relation $X\geq Y$ means that $X-Y$ is positive semi-definite, and the relation $X \gg Y$ means that the support of $Y$ is contained in the support of~$X$. Throughout this paper we denote by $\log$ the logarithm to base $2$.

We denote the \emph{trace distance} between two states $\rho$ and $\sigma$ by $T(\rho,\sigma) := \frac{1}{2} \norm{\rho - \sigma}_1$, where $\| \cdot \|_1$ denotes the Schatten 1-norm. The \emph{fidelity} is given as $F(\rho,\sigma) := \| \sqrt{\rho}\sqrt{\sigma}\|_1^2$. We will also use the \emph{sine} or \emph{purified distance}, $P(\rho,\sigma) := \sqrt{1-F(\rho,\sigma)}$. Both $P$ and $T$ satisfy the triangle inequality and are non-increasing when a CPTP map is applied to both states. They are related by the Fuchs-van de Graaf inequalities stating that $T(\rho,\sigma) \leq P(\rho,\sigma) \leq \sqrt{2T(\rho,\sigma) - T(\rho,\sigma)^2}$. When the choice of metric is arbitrary, we use $\Delta$ to stand for either $T$ or $P$.

\subsection{Some divergences and their properties}

In this work we use several different non-commutative divergences. We will introduce here only the measures and properties that are needed for this work\,---\,an interested reader may consult~\cite{mybook} for a more comprehensive discussion with references to all the original papers.
For $\rho, \sigma \in \St{d}$ the relative entropy is given by $D(\rho \| \sigma):= \tr\, \rho(\log \rho - \log \sigma)$ if $\sigma \gg \rho$, and $+\infty$ otherwise. To simplify the exposition in the following we assume throughout that the states $\sigma$ always have full support and are thus invertible, avoiding such infinities.

For $\alpha \in (0,1) \cup (1,2]$ the \emph{Petz quantum R\'enyi divergence}~\cite{petz86} is defined as
\begin{align}
\bar D_{\alpha}(\rho \| \sigma):=
\frac{1}{\alpha -1 } \log \tr\, \rho^\alpha \sigma^{1-\alpha} \,.
\end{align}
In the limit $\alpha \to 0$ the divergence converges to the \emph{min-relative entropy}~\cite{renner05,datta08}, i.e.,
\begin{align}
   D_{\min}(\rho\|\sigma) &:= -\log \tr \,\sigma \Pi_{\rho} = \lim_{\alpha \to 0} \bar D_{\alpha}(\rho \| \sigma)  \, ,
\end{align}
where $\Pi_{\rho}$ is the projector onto the support of $\rho$.

Another non-commutative family of R\'enyi divergences is the \emph{sandwiched quantum R\'enyi divergence}~\cite{lennert13,wilde13}, for $\alpha \in [\frac12,1) \cup (1,\infty)$ defined as
\begin{align}
\tilde D_{\alpha}(\rho \| \sigma):= \frac{1}{\alpha -1 } \log \tr \big( \sigma^{\frac{1-\alpha}{2 \alpha}} \rho \sigma^{\frac{1-\alpha}{2 \alpha}} \big)^{\alpha}  \, .
\end{align}
For $\alpha = \frac{1}{2}$ the sandwiched quantum R\'enyi divergence becomes
$\tilde D_{\nicefrac{1}{2}}(\rho \| \sigma) = - \log F(\rho,\sigma)$.
In the limit $\alpha \to \infty$ the sandwiched quantum R\'enyi divergence converges to the so-called \emph{max-divergence}~\cite{renner05,datta08}, i.e.,  
\begin{align}
 D_{\max}(\rho\|\sigma) &:= \inf \big\{ \lambda \in \mathbb{R} : \rho \leq 2^{\lambda} \sigma \big\} \, .
\end{align}
Both non-commutative families of R\'enyi divergences introduced above satisfy many desirable properties: they are monotonically increasing in $\alpha$, they satisfiy the data-processing inequality (i.e.\ they are non-increasing when the same CPTP map is applied to both states), and in the limit $\alpha \to 1$ they both converge to the relative entropy. 

Smooth variants of the max and min-divergence will be useful to treat problems with finite errors. The \emph{$\eps$-smooth max-divergence} is defined as
\begin{align}
D^{\eps,\Delta}_{\max}(\rho\|\sigma):= \inf_{\tilde \rho \in \cB_{\eps}^{\Delta}(\rho)} D_{\max}(\tilde \rho\|\sigma) \, ,
\end{align}
where $\cB_{\eps}^{\Delta}(\rho):= \{\tilde \rho \in \St{d}: \Delta(\rho,\tilde \rho) \leq \eps \}$ for $\eps \in (0,1)$.  We will use this definition for both trace distance and purified distance, but note that in contrast to some other works the optimization here only goes over close quantum states, not sub-normalized states.
As the smooth variant of $D_{\min}$ we introduce the so-called \emph{hypothesis testing divergence}.\footnote{In some earlier works (see, e.g.,~\cite{datta08}) an optimization over a ball of $\eps$-close states is used instead to define a smooth variant of $D_{\min}$, in analogy with $D_{\max}^{\eps}$ defined above. However, the quantity we use here, first introduced in~\cite{datta09} as a generalization of $D_{\min}$, is more natural for our problem here.}
For any $\eps \in (0, 1)$ it is defined as
\begin{align} \label{eq_Dh}
  D_h^{\eps}(\rho\|\sigma) :=  -\log \inf \big\{ \tr\, \sigma Q : 0 \leq Q \leq \id \ \land \ \tr\, \rho \, Q \geq 1 - \eps \big\} \,.
\end{align} 
In the limit $\eps \to 0$ we recover $D_{\min}(\rho\|\sigma)$. Finally, we note that both of these divergences satisfy the data-processing inequality, namely
\begin{align}
	D^{\eps,\Delta}_{\max}(\rho\|\sigma) \geq D^{\eps,\Delta}_{\max}( \cE(\rho) \| \cE(\sigma)) 
	\qquad \textrm{and} \qquad
	D^{\eps}_{h}(\rho\|\sigma) \geq D^{\eps}_{h}( \cE(\rho) \| \cE(\sigma)) 
\end{align}
for any CPTP map $\cE$.

The smooth max-divergence and the hypothesis testing divergence are closely related, as shown in~\cite[Theorem~4]{anshu19}.
\begin{proposition}
	\label{prop:DhDmax}
	Let $\rho, \sigma \in \St{d}$ and $\eps \in (0,1)$ and $\nu \in (0,1-\eps)$. It holds that
\begin{align}
 	D_h^{1-\eps}(\rho\|\sigma)  \geq D_{\max}^{\sqrt{\eps},P}(\rho\|\sigma) - \log\frac{1}{1-\eps}  \geq D_h^{1-\eps-\nu}(\rho\|\sigma) - \log\frac{4}{\nu^2}  \,.
\end{align}
\end{proposition} 

We next recall two inequalities that relate the hypothesis testing relative entropy and the max-relative entropy to the R\'enyi divergences.
\begin{proposition} \label{prop:renyi}
Let $\eps \in (0,1)$ and let $\rho, \sigma \in \St{d}$. Then,
\begin{align} \label{eq_DhRenyi}
D_{h}^\eps(\rho \| \sigma) &\geq \bar D_{\alpha}(\rho \| \sigma) - \frac{\alpha}{1-\alpha } \log \frac{1}{\eps}  && \textrm{for } \alpha \in [0,1) \, , \quad \textrm{and} \\
\label{eq_DmaxRenyi}
D_{\max}^{\eps,\Delta}(\rho \| \sigma) &\leq \tilde D_{\alpha}(\rho \| \sigma) + \frac{1}{\alpha-1} \log \frac{1}{\eps^2} + \log \frac{1}{1-\eps^2} && \textrm{for } \alpha \in (1,\infty] \, .
\end{align}
We can interchange $\bar D_{\alpha}$ and $\tilde D_{\alpha}$ in the above inequalities.
\end{proposition}
Note that we can interchange the R\'enyi relative entropies in both inequalities since $\bar D_{\alpha}(\rho\|\sigma) \geq \tilde D_{\alpha}(\rho\|\sigma)$ for all states due to the Araki-Lieb-Thirring inequality.
Inequality~\eqref{eq_DhRenyi} follows immediately from~\cite[Proposition~3.2]{audenaert12}. Inequality~\eqref{eq_DmaxRenyi} is shown in~\cite[Theorem 3]{anshu19}, tightening earlier bounds that were established as part of the fully quantum asymptotic equipartition property (QAEP)~\cite{tomamichel08}. The QAEP states that, for all $\eps \in (0,1)$, the regularized smooth entropies converge to the relative entropy
\begin{align}
	\label{eq:aep}
	\lim_{n \to \infty} \frac{1}{n} D_{\max}^{\eps,\Delta} \left(\rho^{\otimes n} \middle\| \sigma^{\otimes n} \right) = D(\rho\|\sigma) \,.
\end{align}
The analogous statement for $D_h^{\eps}$ is an immediate consequence from quantum Stein's lemma and its converse~\cite{hiai91,ogawa00}.

\section{Conditions for pairwise state transformation}
\label{sec:main}

In this section we derive sufficient conditions for the existence of a channel that transforms $(\rho_1,\sigma_1)$ to $(\rho_2,\sigma_2)$, where the first state is transformed either exactly or approximately, and the second state always has to be transformed exactly.

\subsection{Conditions for exact state transformation}
\label{sec:main-exact}

We start by considering the case of exact transformations. We note that conditions for an exact transformation derived in this section are very restrictive. If we allow for approximate transformations as done in Section~\ref{sec:main-approximate} we will find conditions that are considerably easier to fulfill.

For what follows, we can restrict ourselves to a very special class of transformations, namely, \textit{test-and-prepare channels} of the form
\begin{align}
\mathcal{E}(\rho)=\gamma_1\; \tr\, E\rho + \gamma_2\; \tr\,(\id-E)\rho \;,
\end{align}
where the $\gamma_i$'s are density matrices and 
$0\leq E\leq \id$. Hence test-and-prepare channels constitute a subset of measure-and-prepare channels, in which the measurement is a simple binary test. Ref.~\cite{Buscemi2017} provides a complete characterization of this case. The following lemma can be obtained as a consequence of the results in~\cite{Buscemi2017}, but we provide an independent proof here for the sake of the reader (see also Ref.~\cite{matsumoto2014-cq-transform}).

\begin{lemma}\label{lemma}
	Let $\rho_1, \sigma_1 \in \St{2}$ be commuting qubit quantum states and let $\rho_2, \sigma_2 \in \St{d}$.
	The following two conditions are equivalent:
	\begin{enumerate}[(i)]
		\item  \label{it_1} there exists a CPTP map $\mathcal{E}:\St{2}\to\St{d}$ such that $\mathcal{E}(\rho_1)=\rho_2$ and $\mathcal{E}(\sigma_1)=\sigma_2$;
		\item $D_{\max}(\rho_1\|\sigma_1)\ge D_{\max}(\rho_2\|\sigma_2)$ and $D_{\max}(\sigma_1\|\rho_1)\ge D_{\max}(\sigma_2\|\rho_2)$. \label{it_2}
	\end{enumerate}
\end{lemma}

\begin{proof}
	Since the implication \eqref{it_1}$\implies$\eqref{it_2} is just the data-processing inequality, we only need to prove the reverse implication \eqref{it_2}$\implies$\eqref{it_1}.
	
	By assumption, $\rho_1$ and $\sigma_1$ commute. Hence, we can see them as classical binary probability distributions, namely, $\rho_1 \leftrightarrow \vec{p}_1=(p,1-p)$ and $\sigma_1 \leftrightarrow \vec{q}_1=(q,1-q)$. Moreover, we can assume that $\frac{p}{q} \ge \frac{1-p}{1-q}$; otherwise, the first step is to map $(p,1-p)$ and $(q,1-q)$ into $(1-p,p)$ and $(1-q,q)$, respectively. Notice that this is equivalent to saying that, without loss of generality, we can always assume $p(1-q)\ge q(1-p)$, that is, $p\ge q$.
	
	Let us now define $M := \frac{p}{q}$ and $m := \frac{1-p}{1-q}$. Notice that $\log M=D_{\max}(\rho_1\|\sigma_1)$ and $\log m=-D_{\max}(\sigma_1\|\rho_1)$. Problems with such definitions occur when $q=0$ or $q=1$. If $q=1$, since (as noticed above) we can assume that $p\ge q$, then also $p=1$; that is, $\rho_1=\sigma_1$ and $D_{\max}(\rho_1\|\sigma_1)=-D_{\max}(\sigma_1\|\rho_1)=0$ (or, equivalently, $M=m=1$). By Condition~\eqref{it_2}, this implies that also $D_{\max}(\rho_2\|\sigma_2)=0$, namely, $\rho_2=\sigma_2$, and there is nothing to prove. The case $q=0$ will be addressed separately at the end of the proof. 
	
	For the moment, we thus assume that $\infty>M>1>m\ge0$, that is, $p>q>0$. Condition~\eqref{it_2} then guarantees that $M\sigma_2-\rho_2\ge 0$ and $\rho_2-m\sigma_2\ge 0$. We thus define a linear map $\mathcal{E}:\St{2}\to\St{d}$ as follows:
\begin{align}
\mathcal{E}(\cdot) := \<0|\cdot|0\>\frac{\rho_2-m\sigma_2}{1-m}+\<1|\cdot|1\>\frac{M\sigma_2-\rho_2}{M-1}\label{eq:expr2}\;.
\end{align}
By construction, $\mathcal{E}$ is clearly CPTP. On the other hand, by direct inspection, 
	$\mathcal{E}(\rho_1)=\rho_2$ and $\mathcal{E}(\sigma_1)=\sigma_2$.
	
	In order to conclude the proof, we are left to address the case in which $q=0$. In such a case, we consider a variant of the map in~(\ref{eq:expr2}), that is,
	\begin{align*}
	\mathcal{F}(\cdot) := \<0|\cdot|0\>\frac{\rho_2-m\sigma_2}{1-m}+\<1|\cdot|1\>\sigma_2\;.
	\end{align*}
	Again, by direct inspection it is easy to check that $\mathcal{F}$ is CPTP and that $\mathcal{E}(\rho_1)=\rho_2$ and $\mathcal{E}(\sigma_1)=\sigma_2$, as claimed.
\end{proof}

From the above, we obtain a sufficient condition for the existence of a transformation, more precisely a test-and-prepare channel, for arbitrary pairs of states.

\begin{corollary}\label{coro}
	Let $\rho_1, \sigma_1 \in \St{d_1}$ and $\rho_2, \sigma_2 \in \St{d_2}$. Then, if either
	\begin{align}\label{eq_suffCond}
	D_{\min}(\rho_1\|\sigma_1)\ge D_{\max}(\rho_2\|\sigma_2) \qquad \textrm{or} \qquad D_{\min}(\sigma_1\|\rho_1)\ge D_{\max}(\sigma_2\|\rho_2) \,,
	\end{align}
	then there exists a test-and-prepare channel $\mathcal{E}$ such that $\mathcal{E}(\rho_1)=\rho_2$ and $\mathcal{E}(\sigma_1)=\sigma_2$. 
\end{corollary}

\begin{proof}
	It suffices to show the statement under the first condition; the second then follows by symmetry.
	Let us consider the measurement channel
	\[
	\mathcal{M}(\cdot) := \Tr[\cdot\;\Pi]|0\>\<0|+\Tr[\cdot\;(\id-\Pi)]|1\>\<1|\;,
	\]
	where $\Pi$ is the projector onto the support of $\rho_1$. Then, the binary classical probability distributions obtained from the pair $(\rho_1,\sigma_1)$ are $\mathcal{M}(\rho_1) \leftrightarrow \vec{p}_1 = (p,1-p)=(1,0)$ and $\mathcal{M}(\sigma_1) \leftrightarrow \vec{q}_1= (q,1-q)=(\Tr[\Pi\;\sigma_1],\Tr[(\id-\Pi)\;\sigma_1])$.
	
	We only need to show that, if Eq.~(\ref{eq_suffCond}) holds, then $\vec{p}_1$ and $\vec{q}_1$ satisfy condition~\eqref{it_2} in Lemma~\ref{lemma}. That is indeed the case, since, on the one hand,
	\[
	\frac{p}{q}=\frac{1}{\Tr[\Pi\;\sigma_1]} =  2^{D_{\min}(\rho_1\|\sigma_1)}\;,
	\]
	that is, $D_{\max}(\vec{p}_1\|\vec{q}_1)=D_{\min}(\rho_1\|\sigma_1)$, so that Eq.~(\ref{eq_suffCond}) guarantees that $D_{\max}(\vec{p}_1\|\vec{q}_1)\ge D_{\max}(\rho_2\|\sigma_2)$. Notice that, in the above equation, the case $q=0$ can be excluded, as it is equivalent to $\rho_1$ and $\sigma_1$ being perfectly distinguishable, and thus exactly transformable into any other pair of states.
	
	On the other hand, because
	\[
	2^{-D_{\max}(\vec{q}_1\|\vec{p}_1)}=\frac{1-p}{1-q} = 0 \le 2^{-D_{\max}(\sigma_2\|\rho_2)}\;.
	\]
	 the remaining condition, $D_{\max}(\vec{q}_1\|\vec{p}_1)\ge D_{\max}(\sigma_2\|\rho_2)$, is guaranteed by definition. In the above equation, the case $q=1$ can be excluded, as it is equivalent to $\rho_1$ and $\sigma_1$ having the same support, that is, $D_{\min}(\rho_1\|\sigma_1)=D_{\min}(\sigma_1\|\rho_1)=0$, which, by virtue of assumption~(\ref{eq_suffCond}), implies that $\rho_2=\sigma_2$, so that there is nothing to prove.
\end{proof}

%

We note that condition~\eqref{eq_suffCond} is very strong in general. To see this, it is enough to simply consider two states with coinciding support, which is what happens generically (in the generic case, states are invertible with probability one). Then, since in such a case $D_{\min}(\rho_1 \| \sigma_1)=D_{\min}(\sigma_1 \| \rho_1)=0$, the only dichotomies satisfying Eq.~(\ref{eq_suffCond}) are the trivial ones, i.e., those with $\rho_2=\sigma_2$, as a consequence of the fact that max-relative entropy vanishes if and only if the two arguments coincide.

However, when errors are allowed in condition~(\ref{eq_suffCond}), in what follows we show that a much more flexible framework can be constructed, which is in fact sufficient to derive strong asymptotic results.

\subsection{Sufficient condition for approximate state transformation}
\label{sec:main-approximate}

In this section we are interested in approximate state transformation, i.e., a transformation from $(\rho_1, \sigma_1)$ to $(\rho_2,\sigma_2)$ where we allow for a (small) error in the transformation $\rho_1 \to \rho_2$, while the transformation $\sigma_1 \to \sigma_2$ is required to be exact.

\begin{lemma} \label{lem:smooth}
 Let $\eps_1, \eps_2 \in (0,1)$, $\rho_1, \sigma_1 \in \St{d_1}$ and $\rho_2, \sigma_2 \in \St{d_2}$. If either
\begin{align} \label{eq_assProp}
 D_h^{\eps_1}(\rho_1\|\sigma_1) \geq D_{\max}^{\eps_2, T}(\rho_2\|\sigma_2) \qquad \textrm{or} \qquad
 D_h^{\eps_1}(\rho_1\|\sigma_1) \geq D_{\max}^{\eps_2, P}(\rho_2\|\sigma_2)  \,,
\end{align}
then there exists a test-and-prepare quantum channel $\mathcal{E}: \St{d_1} \to \St{d_2}$ satisfying $\mathcal{E}(\sigma_1) = \sigma_2$ and $T( \mathcal{E}(\rho_1),\rho_2) \leq \eps_1+\eps_2$ or $P( \mathcal{E}(\rho_1),\rho_2) \leq \sqrt{\eps_1}+\eps_2$, respectively.
\end{lemma}

\begin{proof}
Let $Q^*$ denote the optimizer in~\eqref{eq_Dh} for $D_h^{\eps_1}(\rho_1\|\sigma_1)$ that satisfies
\begin{align} \label{eq_dhQ}
 2^{-D_h^{\eps_1}(\rho_1\|\sigma_1)}= \tr\, \sigma_1 Q^* \qquad \textrm{and} \qquad
 \tr \sigma_1 Q^* \leq \tr \, \rho_1 Q^* = 1 - \eps_1 < 1 \, .
\end{align} 
By definition of the smooth max-relative entropy, there furthermore exists a state $\tilde \rho_2 \in \cB_{\eps_2}^{\Delta}(\rho_2)$ for $\Delta \in \{T, P\}$, such that
\begin{align}
D_{\max}^{\eps_2,\Delta}(\rho_2\|\sigma_2) =  D_{\max}(\tilde \rho_2 \| \sigma_2) \, .
\end{align}

Consider now the mapping
\begin{align}
\cE \,:  \,X \mapsto  \cE(X) &= \tilde \rho_2 \tr(X Q^*) +\tr \big(X (\id - Q^*)\big) \frac{\sigma_2 - \tilde \rho_2 \tr(\sigma_1 Q^*)}{\tr \big(\sigma_1 (\id - Q^*)\big)} \\
&= \left( \tr X - \frac{\tr \big(X (\id - Q^*)\big)}{\tr \big(\sigma_1 (\id - Q^*)\big)}\right) \tilde \rho_2 + \left( \frac{\tr \big(X (\id - Q^*)\big)}{\tr \big(\sigma_1 (\id - Q^*)\big)} \right) \sigma_2 \, . \label{eq:themapping2}
\end{align}
We start by proving that $\cE$ is a quantum channel, i.e., a trace-preserving completely positive map. To see that $\cE$ is trace-preserving is straightforward since $\sigma_2$ and $\tilde \rho_2$ are density operators. We note that because $0 \leq Q^* \leq \id$ it suffices to show that $\sigma_2 \geq \tilde \rho_2 \tr(\sigma_1 Q^*)$ in order prove that $\cE$ is completely positive. By definition of the smooth max-relative entropy and by using~\eqref{eq_assProp} and~\eqref{eq_dhQ} we find
\begin{align} \label{eq_fromAss}
2^{-D_{\max}(\tilde \rho_2 \| \sigma_2)} 
= 2^{-D_{\max}^{\eps_2,\Delta}(\rho_2 \| \sigma_2)}  
\geq 2^{-D_h^{\eps_1}(\rho_1 \| \sigma_1)}
= \tr \sigma_1 Q^* \, . 
\end{align}
By definition of the max-relative entropy we thus have
\begin{align} \label{eq_DS1}
\sigma_2 \geq \tilde \rho_2 \, 2^{- D_{\max}(\tilde \rho_2 \| \sigma_2)} \geq  \tilde \rho_2 \, \tr(\sigma_1 Q^{*}) \, .
\end{align}

We have seen that $\cE$ is indeed a quantum channel. It thus remains to show that $\mathcal{E}(\sigma_1) = \sigma_2$ and $\Delta( \mathcal{E}(\rho_1), \rho_2) \leq \eps_1 + \eps_2$. The first property is straightforward to verify. 
The second property requires some more work. 
Define $\bar \rho_2 := \cE(\rho_1)$. The triangle inequality immediately yields
\begin{align}
	\Delta(\rho_2, \bar \rho_2) \leq \eps_2 + \Delta(\tilde \rho_2, \bar \rho_2) \,, \label{eq_almostDone}
\end{align}
and it thus remains to bound the second term. 

Let us first consider the case where $\Delta(\rho,\tau) = T(\rho, \tau) = \frac12\| \rho - \tau \|_1$ is the trace distance. Substituting the expression in~\eqref{eq:themapping2}, we find
\begin{align}
\norm{\tilde \rho_2 - \bar \rho_2 }_1 = \frac{1- \tr \rho_1 Q^*}{1 - \tr \sigma_1 Q^*} \| \tilde \rho_2 - \sigma_2 \|_1 \leq \eps_1 \frac{\| \tilde \rho_2 - \sigma_2 \|_1}{1-2^{-D_{\max}(\tilde \rho_2 \| \sigma_2)}} \,,
\end{align}
where we used~\eqref{eq_dhQ} and~\eqref{eq_fromAss} in the final step. Using~\eqref{eq_DS1} we find for $\Pi_+$ denoting the projector on to the positive support of $\rho_2 - \sigma_2$
\begin{align}
\norm{\tilde \rho_2 - \sigma_2}_1
= 2\, \tr\, \Pi_+(\tilde \rho_2 - \sigma_2)
&\leq 2\Big(1-2^{-D_{\max}(\tilde \rho_2 \| \sigma_2)}\Big) \tr\, \Pi_+ \tilde \rho_2 \\
&\leq 2\Big(1-2^{-D_{\max}(\tilde \rho_2 \| \sigma_2)}\Big) \, .
\end{align}
Combining this with~\eqref{eq_almostDone} finally gives
$\Delta\big(\cE(\rho_1),\rho_2\big) \leq \eps_1 + \eps_2$, concluding the proof for the trace distance.

Let us now consider the case where $\Delta(\rho,\tau) = P(\rho,\tau) = \sqrt{1- F(\rho,\tau)}$ is the purified distance.
First, we recall that $\tr \rho_1 Q^* \geq \tr \sigma_1 Q^*$ by definition of $Q^*$ as the optimizer for the hypothesis testing divergence.
We then use the concavity of fidelity (see, e.g.,~\cite[Lemma 3.4]{mybook}) and the expression in~\eqref{eq:themapping2} to bound
\begin{align}
	F(\tilde \rho_2, \bar \rho_2) &\geq \left( 1 - \frac{1- \tr \rho_1 Q^*}{1 - \tr \sigma_1 Q^*}\right) + \left( \frac{1- \tr \rho_1 Q^*}{1 - \tr \sigma_1 Q^*} \right) F(\tilde \rho_2, \sigma_2) \\
	&= 1 - \frac{1- \tr \rho_1 Q^*}{1 - \tr \sigma_1 Q^*} \left( 1 - F(\tilde \rho_2, \sigma_2) \right) \\
	&\geq 1 - \eps_1 \frac{1 - F(\tilde \rho_2, \sigma_2)}{1-2^{-D_{\max}(\tilde \rho_2 \| \sigma_2)}} \,.
\end{align}
Using the monotonicity of sandwiched R\'enyi relative entropy we find $D_{\max}(\tilde \rho_2\|\sigma_2) \geq -\log F(\tilde \rho_2, \sigma_2)$ and $F(\tilde \rho_2, \bar \rho_2) \geq 1 - \eps_1$, concluding that $\Delta\big(\cE(\rho_1),\rho_2\big) \leq \sqrt{\eps_1} + \eps_2$.
\end{proof}

\subsection{Conditions for asymptotic state transformation}
\label{sec:main-asymptotic}

In the following we will consider an asymptotic setting given by four sequences of states $\vec{\rho}_1 = \{\rho_1^n\}_n$, ${\vec{\sigma}}_1 = \{\sigma_1^n\}_n$, ${\vec{\rho}}_2 = \{\rho_2^n\}_n$, and ${\vec{\sigma}}_2 = \{\sigma_2^n\}_n$ for $n \in \N$. We assume no specific structure for these states and the underlying Hilbert spaces, i.e.\ in general we have $\rho_1^n, \sigma_1^n \in \St{d_1^n}$ and $\rho_2^n, \sigma_2^n \in \St{d_2^n}$ for arbitrary dimensions $\{d_1^n\}_n$ and $\{d_2^n\}_n$.  
The only requirement that we impose is that, for $i \in \{1, 2\}$, the limits
\begin{align}
	\tilde D_{\alpha}( \vec{\rho}_i \| \vec{\sigma}_i ) := \lim_{n \to \infty} \frac{1}{n} \tilde D_{\alpha}(\rho_i^n \| \sigma_i^n) \label{eq:condition}
\end{align}
exist and are continuous in $\alpha$ at $\alpha = 1$. To simplify notation we write $D(\vec{\rho}_i \| \vec{\sigma}_i )$ := $\tilde D_{1}(\vec{\rho}_i \| \vec{\sigma}_i )$.

A simple sequence satisfying this condition is given by independent and identically distributed (iid) states, i.e.\ the sequences determined by $\rho_i^n = \rho_i^{\otimes n}$ and $\sigma_i^n = \sigma_i^{\otimes n}$ for $i \in \{1, 2\}$. In this case the expressions simplify due to the additivity of the R\'enyi divergence for product states and we have $\tilde D_{\alpha}(\vec{\rho}_i \| \vec{\sigma}_i ) = \tilde D_{\alpha}(\rho_i\|\sigma_i)$ as well as $D(\vec{\rho}_i \| \vec{\sigma}_i ) = D(\rho_i\|\sigma_i)$. A slightly more involved example is constructed by taking tensor products of states, each randomly chosen from some finite sets. The resulting sequences of product states satisfy $\tilde D_{\alpha}(\vec{\rho}_i \| \vec{\sigma}_i ) = \mathbb{E} [ \tilde D_{\alpha}(\rho_i\|\sigma_i) ]$ by the law of large numbers, where the expectation is taken over the joint distribution of $\rho_1, \sigma_1, \rho_2$ and $\sigma_2$.

In the following we show that sufficient and necessary conditions for asymptotic state transformations are determined by $\lambda_1 := D(\vec{\rho}_1 \| \vec{\sigma}_1 )$ and $\lambda_2 := D(\vec{\rho}_2 \| \vec{\sigma}_2 )$. 
\begin{itemize}
	\item If $\lambda_1 > \lambda_2$ we show the existence of a sequence of channels that transform $(\rho_1^n,\sigma_1^n)$ to $(\rho_2^n,\sigma_2^n)$ where the transformation $\rho_1^n \to \rho_2^n$ has an error that is vanishing exponentially as $n \to \infty$ and the transformation $\sigma_1^n \to \sigma_2^n$ is exact. 
	\item On the other hand, if $\lambda_1 < \lambda_2$ we show that any  transformation for which $\sigma_1^n \to \sigma_2^n$ is exact leads to an error exponentially approaching one as $n \to \infty$ in the transformation $\rho_1^n \to \rho_2^n$. 
\end{itemize}
We note that the case where $\lambda_1 = \lambda_2$ is left as an open question. For example, in case of four states $\rho_1,\rho_2,\sigma_1,$ and $\sigma_2$ such that $D(\rho_1 \| \sigma_1) = D(\rho_2 \| \sigma_2)$ it is unknown if there exists a sequence of channels that take $\sigma_1^{\otimes n}$ to $\sigma_2^{\otimes n}$ for each $n$ and that take $\rho_1^{\otimes n}$ to $\rho_2^{\otimes n}$ up to asymptotically vanishing error.

Our main technical results are formally presented in the next two theorems.

\begin{theorem}[Achievability with exponentially small error]
  \label{th:main}
Let $\vec{\rho}_1$, $\vec{\sigma}_1$, $\vec{\rho}_2$ and $\vec{\sigma}_2$ for $n \in \N$ be sequences satisfying the condition in~\eqref{eq:condition} and furthermore
   \begin{align}
		D(\vec{\rho}_1 \| \vec{\sigma}_1 ) > D(\vec{\rho}_2 \| \vec{\sigma}_2 ) \,.
   \end{align}
   Then there exists $\gamma>0$, $n_0 \in \N$, and sequence $\{\cE_n\}_{n \in \N}$ of (test-and-prepare) quantum channels such that 
	\begin{align}
	  \cE_n(\sigma_1^n)= \sigma_2^n  \quad \forall n \in \N
		\qquad \textnormal{and} \qquad
	  \Delta\big( \cE_n(\rho_1^n), \rho_2^n \big) \leq 2^{-\gamma n} \quad \forall n\geq n_0  \, .
	\end{align}
\end{theorem}


\begin{remark}
	The above theorem strengthens~\cite[Theorem~2.7]{matsumoto10} in that we are also able to show the exponential decay of the error in the transformation. Furthermore, we extend the result~\cite[Theorem~2.7]{matsumoto10} to hold for states that do not necessarily satisfy an iid~assumption. 
\end{remark}

\begin{proof}
	We show the statement for trace distance, and the statement for purified distance then follows by the Fuchs-van de Graaf inequality.

By the assumption of the theorem and the continuity guaranteed by the assumption at the beginning of the section, there exists a $\delta > 0$ and $\kappa > 0$ such that 
\begin{align}
	\tilde D_{1-\delta}(\vec{\rho}_1 \| \vec{\sigma}_1 ) \geq \tilde D_{1+\delta}(\vec{\rho}_2 \| \vec{\sigma}_2 ) + \kappa \,. 
\end{align}
Hence, by their definition as limits, there exists a $n_0 \in \N$ such that for all $n \geq n_0$,
\begin{align}
	\tilde D_{1-\delta}( \rho_1^n \| \sigma_1^n) \geq \tilde D_{1+\delta}(\rho_2^n \| \sigma_2^n) + \frac{n \kappa}{2} .\label{eq:alphacont}
\end{align}

Let us now set $\eps_n = \frac12 2^{-\gamma n}$ for some $\gamma > 0$ to be determined later. From Lemma~\ref{lem:smooth} we learn that the maps $\mathcal{E}_n$ with the desired properties exist if
\begin{align}
	D_h^{\eps_n}(\rho_1^n\|\sigma_1^n) \geq D_{\max}^{\eps_n,T}(\rho_2^n\|\sigma_2^n) \label{eq:ordering} \,.
\end{align}
Indeed, Proposition~\ref{prop:renyi} together with~\eqref{eq:alphacont} imply
\begin{align}
	D_h^{\eps_n}(\rho_1^n\|\sigma_1^n) - D_{\max}^{\eps_n,T}(\rho_2^n\|\sigma_2^n) &\geq \frac{n\kappa}{2} - \frac{1-\delta}{\delta}\log \frac{1}{\eps_n} - \frac{1}{\delta}\log \frac{1}{\eps_n^2} - \log \frac{1}{1-\eps_n^2} \label{eq:penultimate} \\
	&\geq \frac{n \kappa}{2} - \frac{3}{\delta} \log \frac{1}{\eps_n} \label{eq:last}\\
	&\geq \frac{n \kappa}{2} - \frac{4 \gamma n}{\delta} \, ,
\end{align}
where to reach~\eqref{eq:last} we used that $1-\eps_n^2 \geq \eps_n$ since $\eps_n \leq \frac12$, and in the last step we assumed $\gamma n \geq 3$. We conclude that the choice $\gamma = \frac{\kappa\delta}{8}$ ensures that~\eqref{eq:ordering} holds.
\end{proof}

\begin{theorem}[Exponential strong converse]
	\label{th:converse}
	Let $\vec{\rho}_1$, $\vec{\sigma}_1$, $\vec{\rho}_2$ and $\vec{\sigma}_2$ for $n \in \N$ be sequences satisfying the condition in~\eqref{eq:condition} and furthermore
	\begin{align}
		 D(\vec{\rho}_1 \| \vec{\sigma}_1 ) < D(\vec{\rho}_2 \| \vec{\sigma}_2 ) \,.
	\end{align}
Then there exists $\gamma >0$ such that for all sequences of quantum channels $\{\cE_n\}_{n \in \N}$ that satisfy $\cE_n(\sigma_1^n) = \sigma_2^n$ there exists an $n_0 \in \N$ such that for all $n\geq n_0$ we have
\begin{align}
\Delta\big( \cE_n(\rho_1^n), \rho_2^n \big) \geq 1 - 2^{- \gamma n} \, .
\end{align}
\end{theorem}

\begin{proof}
	We show the statement for purified distance, and the statement for trace distance then follows by the Fuchs-van de Graaf inequality.

We again start by observing that the assumption of the theorem and the continuity guaranteed by the assumption at the beginning of the section imply the existence of a $\delta > 0$ and $\kappa > 0$ such that $\tilde D_{1+\delta}(\vec{\rho}_1 \| \vec{\sigma}_1 ) \leq \tilde D_{1-\delta}( \vec{\rho}_2 \| \vec{\sigma}_2) - \kappa$, and thus, there exists a $n_0 \in \N$ such that for all $n \geq n_0$, we have
\begin{align}
	\tilde D_{1+\delta}( \rho_1^n \| \sigma_1^n) - \tilde D_{1-\delta}(\rho_2^n \| \sigma_2^n) \leq - \frac{n \kappa}{2} .\label{eq:alphacont2}
\end{align}

It thus suffices to prove that~\eqref{eq:alphacont2} implies the desired property for all sequences of quantum channels. In the following we prove the contrapositive. Suppose that for all $\gamma > 0$, there exists a family of channels $\{\cE_n\}_{n \in \N}$ such that for some $n \geq n_0$ we have both $\cE_n(\sigma_1^n)= \sigma_2^n$ and  $P(\cE_n(\rho_1^n), \rho_2^n) < 1 - 2^{- \gamma n}$. Let us then fix a $\gamma$, to be determined later, and set $\eps_n = \frac12 2^{-\gamma n}$.
By Proposition~\ref{prop:renyi} we have
\begin{align}
\tilde D_{1+\delta} (\rho_1^n \| \sigma_1^n)
&\geq D_{\max}^{\eps_n,P}(\rho_1^n \| \sigma_1^n) - \frac{1}{\delta} \log \frac{1}{\eps_n^2} - \log \frac{1}{1-\eps_n^2} \\
& \geq D_{\max}^{\eps_n,P} \big( \mathcal{E}_n(\rho_1^n ) \big\| \sigma_2^n\big) - \frac{1}{\delta} \log \frac{1}{\eps_n^2} - \log \frac{1}{\eps_n} \\
& \geq D_{\max}^{1-\eps_n,P}(\rho_2^n \| \sigma_2^n) - \frac{1}{\delta} \log \frac{1}{\eps_n^2} - \log \frac{1}{\eps_n} \label{eq:comb1}
\end{align}
where the penultimate step uses the data-processing inequality for the smooth max-divergence and the fact that $1 - \eps_n^2 \geq \eps_n$ since $\eps_n \leq \frac12$. The final step follows from the definition of the smooth max-relative entropy as an optimization over a ball of close states and the triangle inequality of the purified distance. 

Instantiating Proposition~\ref{prop:DhDmax} with $\eps = (1-\eps_n)^2$ and $\nu = \eps_n - \eps_n^2$ further yields
\begin{align}
	D_{\max}^{1-\eps_n,P}(\rho_2^n \| \sigma_2^n) &\geq D_h^{\eps_n}(\rho_2^n \| \sigma_2^n) + \log \frac{1}{2\eps_n - \eps_n^2} - 2 \log \frac{2}{\eps_n - \eps_n^2} \\
	&\geq D_h^{\eps_n}(\rho_2^n \| \sigma_2^n) - 2 \log \frac{4}{\eps_n} \\
	&\geq \tilde D_{1-\delta}(\rho_2^n \| \sigma_2^n) - \frac{1-\delta}{\delta} \log \frac{1}{\eps_n} - 2 \log \frac{4}{\eps_n}
	, \label{eq:comb2}
\end{align}
where for the second inequality we used that $\eps_n^2 \leq \frac12 \eps_n$ since $\eps_n \leq \frac12$ and the last inequality is a consequence of Proposition~\ref{prop:renyi}. Combining~\eqref{eq:comb1} and~\eqref{eq:comb2} we find
\begin{align}
	\tilde D_{1+\delta}( \rho_1^n \| \sigma_1^n) - \tilde D_{1-\delta}(\rho_2^n \| \sigma_2^n) &\geq - \frac{3}{\delta} \log \frac{1}{\eps_n} - 2 \log \frac{4}{\eps_n} \\
	&= - n \gamma \left( \frac{3}{\delta} + 2 \right) -\frac{3}{\delta} - 6 \, .
\end{align}
However, this contradicts~\eqref{eq:alphacont2} for sufficiently large $n_0$ (and thus $n$) as long as $\gamma$ chosen small enough so that $\gamma \left( \frac{3}{\delta} + 2 \right) < \frac{\kappa}{2}$.
\end{proof}

It is worth noting that we made no attempts to characterize the exact error exponents and strong converse exponents here. This is because finding the exact error exponent for this problem is still open even in simple commutative cases where, for example, $\sigma_1$ and $\sigma_2$ are proportional to the identity and thus commute with $\rho_1$ and $\rho_2$.



\section{Rates for pairwise state transformations}
\label{sec:rates}

From an information theoretic perspective, a natural question to ask is at what rate we can transform between pairs of states asymptotically. The main result of this section has been suggested already in the work of Matsumoto~\cite{matsumoto10}, and is a consequence of Theorems~\ref{th:main} and~\ref{th:converse} of the previous section. It has been obtained independently in concurrent work~\cite{wang19}.
Hence, in this section we ask at what rate we can transform a pair of states $(\rho_1, \sigma_1 )$ into $(\rho_2, \sigma_2 )$ up to some asymptotically vanishing error in the first transformation.  To make this precise, let us first say that a triplet $(n, m, \eps)$ for $m, n \in \mathbb{N}$ and $\eps \in [0, 1]$ is achievable by a transformation if and only if there exists a channel $\cE$ such that 
\begin{align} \label{eq_mapProp}
 \cE\big( \sigma_1^{\otimes n} \big) = \sigma_2^{\otimes m} \quad \text{and} \quad \Delta\big( \cE( \rho_1^{\otimes n}),\rho_2^{\otimes m} \big) \leq \eps \, .
\end{align}
We can now define the maximal achievable pairwise state transformation rate with error $\eps$ on input block length $n$ as 
\begin{align}
\hat R_{\rho_1,\sigma_1 \to \rho_2, \sigma_2}^{\,\eps,\Delta}(n):=\max \left\{ \frac{m}{n} \, : \, (n, m, \eps) \text{ is achievable} \right\} \, .
\end{align}

Our goal is to understand the asymptotics of this quantity for $n \to \infty$ when $\eps$ is constant. We can determine the first order asymptotics of $\hat R_{\rho_1,\sigma_1 \to \rho_2, \sigma_2}^{\,\eps,\Delta}(n)$, which turns out to be independent of the metric $\Delta \in \{T, P\}$ and~$\eps$.

\begin{theorem}
	\label{th:rates}
  Let $\rho_1, \sigma_1 \in \St{d_1}$ and $\rho_2, \sigma_2 \in \St{d_2}$ be two pairs of states.
  For all $\eps \in (0,1)$, the pairwise state transformation rate satisfies
  \begin{align}\label{eq_rate}
 \lim_{n \to \infty} \hat R_{\rho_1,\sigma_1 \to \rho_2, \sigma_2}^{\,\eps,\Delta}(n) = \frac{D(\rho_1\|\sigma_1)}{D(\rho_2\|\sigma_2)} \,.
\end{align}
\end{theorem}

\begin{remark}
	Another viewpoint on this question can be taken by fixing a rate $R$ and asking how the minimal achievable error $\eps$ behaves as a function of $n$.
	Strong qualitative statements for when we transform above and below the critical rate $D(\rho_1\|\sigma_1)/D(\rho_2\|\sigma_2)$ are immediate from Theorems~\ref{th:main} and~\ref{th:converse}. Namely, for transformations below the critical rate the error will drop exponentially as $n \to \infty$, and for transformations above the critical rate the error will approach one exponentially fast as $n \to \infty$.
\end{remark}

\begin{proof}
First consider the case $D(\rho_2\|\sigma_2) = 0$. This implies that $\rho_2 = \sigma_2$ and we can  employ a sequence of constant output channels, $\mathcal{E}_n(\cdot) = \sigma_2^{\otimes m}$, for any $m \in \mathbb{N}$, which implies $\hat R_{\rho_1,\sigma_1 \to \rho_2, \sigma_2}^{\,\eps,\Delta}(n) = \infty$.

 In case $D(\rho_2\|\sigma_2) > 0$ we first show that the rate~\eqref{eq_rate} can be achieved. 
 %
 Let $R = \frac{p}{q} >0$ with $p, q \in \mathbb{N}$ with $D(\rho_1\|\sigma_1) > R D(\rho_2\|\sigma_2)$. The additivity of the relative entropy for tensor product states yields
	\begin{align}
	  D\big(\rho_1^{\otimes q} \big\|\sigma_1^{\otimes q}\big) > D\big(\rho_2^{\otimes p } \big\|\sigma_2^{\otimes p} \big)\,.
	\end{align}
	Hence, Theorem~\ref{th:main} implies the existence of a sequence of channels $\{ \mathcal{E}_{n} \}_{n \in \mathbb{N}}$ such that
\begin{align}
\cE_{n}(\sigma_1^{\otimes n q }) = \sigma_2^{\otimes R n q}   \qquad \text{and} \qquad  \Delta \big( \cE_{n}(\rho_1^{\otimes nq}) , \rho_2^{\otimes R n q}\big) \leq \eps \quad \text{ for large enough $n$ \, . }
\end{align}
We thus showed that $(nq, Rnq, \eps)$ is achievable. Moreover, by just throwing away $s \in \{1, 2, \ldots, q-1\}$ systems, we also know that $(nq+s, Rnq, \eps)$ is achievable. Hence,
\begin{align*}
	\hat R_{\rho_1,\sigma_1 \to \rho_2, \sigma_2}^{\,\eps,\Delta}(nq + s) \geq R \cdot \frac{ nq }{nq + s} \, ,
\end{align*}
and, thus, $\liminf_{N \to \infty}  \hat R_{\rho_1,\sigma_1 \to \rho_2, \sigma_2}^{\,\eps,\Delta}(N) \geq R$. Hence the achievability statement follows because $\mathbb{Q}$ is dense in $\mathbb{R}$.

We next prove a strong converse, i.e., we show that if $\eps$ is bounded away from $1$ then for large $n$ we must have $D(\rho_1\|\sigma_1) \geq R D(\rho_2\|\sigma_2)$. Let $\{\cE_n\}_{n\in\N}$ be a sequence of channels satisfying~\eqref{eq_mapProp}.
The fully quantum asymptotic equipartition property~\cite{tomamichel08,mybook} states that for any $\delta \in (0,1-\eps)$ we have
\begin{align}
D(\rho_1 \| \sigma_1) 
&= \lim_{n \to \infty} \frac{1}{n} D_{\max}^{\delta,\Delta}(\rho_1^{\otimes n} \| \sigma_1^{\otimes n})\\
&\geq  \lim_{n \to \infty} \frac{1}{n} D_{\max}^{\delta,\Delta}\big(\cE_n(\rho_1^{\otimes n}) \| \cE_n( \sigma_1^{\otimes n} )\big)\\
&=\lim_{n \to \infty} \frac{1}{n} D_{\max}^{\delta,\Delta}\big(\cE_n(\rho_1^{\otimes n}) \| \sigma_2^{\otimes  \lceil R n \rceil} \big) \, ,
\end{align}
where the penultimate step follows from the data-processing inequality for the smooth max-relative entropy~\cite{datta08}.
Because $\Delta(\cE_n(\rho_1^{\otimes n}), \rho_2^{\otimes  \lceil R n \rceil}) \leq \eps$, we further find 
\begin{align}
\lim_{n \to \infty} \frac{1}{n} D_{\max}^{\delta,\Delta}\big(\cE_n(\rho_1^{\otimes n}) \| \sigma_2^{\otimes  \lceil R n \rceil} \big)
\geq \lim_{n \to \infty} \frac{1}{n} D_{\max}^{\delta+ \eps,\Delta}\big(\rho_2^{\otimes  \lceil R n \rceil} \| \sigma_2^{\otimes  \lceil R n \rceil} \big)
= R D(\rho_2 \| \sigma_2) \, ,
\end{align}
where the final step again uses the asymptotic equipartition property. 
Combining all of this gives	 $D(\rho_1 \| \sigma_1) \geq R D(\rho_2\|\sigma_2)$, concluding the proof.
\end{proof}

\section{Applications to resource theories}
\label{sec:app}

The above results have some immediate consequences in resource theories: in what follows we consider in particular the resource theory of athermality and the resource theory of coherence. In concurrent work, Wang and Wilde~\cite{wang19} also derived some of our results, interpreting them in terms of a resource theory of ``asymmetric distinguishability''. Our perspective is different insofar as we interpret Theorems~\ref{th:main}, \ref{th:converse} and~\ref{th:rates} as building blocks that have applications in different resource theories. 

Let us first consider the resource theory of athermality under Gibbs-preserving maps.  There, we are given a Hamiltonian $E$, an inverse temperature $\beta$ and a Gibbs state $\gamma = \frac{1}{Z} e^{-\beta E}$, where $Z$ is the normalization factor.
One then asks whether there exists a quantum channel that has the Gibbs state $\gamma$ of a quantum system as a fixed point and transforms $\rho_1$ to $\rho_2$. Theorem~\ref{th:main} reveals that there exists a sequence of Gibbs-preserving maps from $\rho_1^{\otimes n}$ to $\rho_2^{\otimes n}$, with exponentially vanishing error as $n \to \infty$, if
\begin{align}
	D(\rho_1\|\gamma) > D(\rho_2\|\gamma) \qquad \textrm{or, equivalently,} \qquad F_H(\rho_1) > F_H(\rho_2) \, , \label{eq:gibbs}
\end{align} 
where we used that the Helmholtz free energy $F_H = U - T H$, where $U$ is the internal energy, $T$ is the temperature and $H$ is the entropy, that is,
\begin{align}
	F_H(\rho) = \tr \rho E + \frac{1}{\beta} \tr \rho \log \rho  = \frac{1}{\beta} \left( D(\rho \| \gamma) - \log Z \right) \,.
\end{align} 
Furthermore, as shown in Theorem~\ref{th:converse}, no such sequence of Gibbs-preserving maps can exist if the inequality in~\eqref{eq:gibbs} is strictly reversed. In fact, any sequence of Gibbs-preserving maps would incur an error approaching one exponentially fast as $n \to \infty$. In concurrent work by Sagawa \emph{et al.}~\cite{sagawa19} similar results are discussed for the resource theory of athermality with thermal operations (which is more restrictive than the Gibbs-preserving maps discussed here) and for sequences of states beyond the iid case. Since thermal operations are Gibbs-preserving, our strong converse results also apply to their setting, showing that there is no asymptotic advantage in allowing Gibbs-preserving maps.

Another resource theory in which Theorem~\ref{th:rates} above plays a role is the resource theory of coherence, in particular, the resource theory of coherence based on dephasing-covariant incoherent operations (DIO)~\cite{Chitambar2016,Marvian2016}. A DIO operation $\mathcal{E}$ is such that its action commutes with the completely dephasing channel $\mathsf{diag}$, that is
\begin{align}
\mathcal{E}\circ\mathsf{diag} = \mathsf{diag}\circ\mathcal{E}\;.
\end{align}
In this framework, the rate at which coherence can be distilled from an initial resource state $\rho$ is defined, as usual, as the optimal rate at which the transformation
\begin{align}
\rho^{\otimes n}\xrightarrow{\operatorname{DIO}}|+\>\<+|^{\otimes m}\;,
\end{align}
where $|+\>=\frac{1}{\sqrt{2}}(|0\>+|1\>)$ is one unit of coherence,
can be achieved with asymptotically vanishing error. Such a rate is known to be equal to $D(\rho\|\mathsf{diag}(\rho))$~\cite{Regula2018,Chitambar2018}. Recently, a relaxation of the DIO paradigm has been proposed and motivated~\cite{rhoDIO2019}: instead of requiring that the transformation $\mathcal{E}$ and the completely dephasing channel $\mathsf{diag}$ commute on all states, the commutation relation is enforced only on the initial resource state $\rho$. In other words, one considers the $\rho$-DIO condition
\begin{align}
(\mathcal{E}\circ\mathsf{diag})(\rho) = (\mathsf{diag}\circ\mathcal{E})(\rho)\;.
\end{align}
The existence of a $\rho$-DIO channel transforming $\rho$ into $\sigma$ can then be easily reformulated as the existence of a channel $\mathcal{E}$ achieving the following mapping of quantum dichotomies:
\begin{align}
(\rho,\mathsf{diag}(\rho))\xrightarrow{\mathcal{E}}(\sigma,\mathsf{diag}(\sigma))\;.
\end{align}
Once formulated in this form, our Theorem~\ref{th:rates} implies that the rate at which coherence can be distilled from $\rho$ by means of $\rho$-DIO operations is given by $D(\rho\|\mathsf{diag}(\rho))$. The asymptotic distillation rate under $\rho$-DIO has been independently computed in Ref.~\cite{rhoDIO2019}. Beyond that, Theorem~\ref{th:converse} establishes an exponential strong converse which implies that if we try to distill at a rate exceeding $D(\rho\|\mathsf{diag}(\rho))$ then the error will go to one exponentially fast.
Interestingly, since the asymptotic distillation rate is the same for both DIO and $\rho$-DIO operations, and since $\rho$-DIO operations constitute a larger set than DIO operations, we have that the above mentioned exponential strong converse property holds for DIO state distillation too. Beyond that, our results also yield an exponential strong converse for general DIO transformations between arbitrary states $\rho$ and $\sigma$ for any rate exceeding the ratio of $D(\rho\|\mathsf{diag}(\rho))$ and $D(\sigma\|\mathsf{diag}(\sigma))$.

\section{Discussion}
\label{sec:discussion}

Given the result of Theorem~\ref{th:rates}, we are left to wonder how quickly $\hat R_{\rho_1,\sigma_1 \to \rho_2, \sigma_2}^{\eps,\Delta}(n)$ approaches the asymptotic limit, or, in other words: what are the higher order terms in the expansion of $\hat R_{\rho_1,\sigma_1 \to \rho_2, \sigma_2}^{\eps,\Delta}(n)$ for large $n$?
Such questions have recently attracted a lot of interest in both classical and quantum information theory. In the quantum setting, the first tight results in this direction were achieved for hypothesis testing~\cite{li12,tomamichel12}:
\begin{align}
	D_h^{\eps}(\rho^{\otimes n}\|\sigma^{\otimes n}) = n D(\rho\|\sigma) + \sqrt{n V(\rho\|\sigma)} \Phi^{-1}(\eps) + O(\log n) ,
\end{align}
where $V(\rho\|\sigma) := \tr \rho \left( \log \rho - \log \sigma \right)^2 - D(\rho\|\sigma)^2$ is the \emph{relative entropy variance}, i.e.\ a quantum generalization of the variance of the log-likelihood ratio, and $\Phi(\cdot)$ is the cumulative normal distribution function. This is the constant error regime of quantum hypothesis testing.
Similar second-order expansions can be derived for the case where $\eps$ is not constant but approaches $0$ and $1$ slower than exponentially, the so-called moderate deviation regime~\cite{chubb17,cheng17}.

Our goal is to establish similar asymptotic expansions for $\hat R_{\rho_1,\sigma_1 \to \rho_2, \sigma_2}^{\eps,\Delta}(n)$. While this is ultimately beyond the scope of this work, we will justify the following conjecture.

\begin{conjecture}
	Let $\rho_1, \sigma_1 \in \St{d_1}$ and $\rho_2, \sigma_2 \in \St{d_2}$ such that
	\begin{align}
		\nu := \frac{D(\rho_1\|\sigma_1)}{V(\rho_1\|\sigma_1)} \cdot \frac{V(\rho_2\|\sigma_2)}{D(\rho_2\|\sigma_2)} \,.
	\end{align}
	is finite, i.e.\ $V(\rho_1\|\sigma_1) > 0$ and $D(\rho_2\|\sigma_2) > 0$. Then, for any $\eps \in (0,1)$, we have
	\begin{align}
		\lim_{n \to \infty}  \sqrt{n} \left( \hat R_{\rho_1,\sigma_1 \to \rho_2, \sigma_2}^{\eps,P}(n) - \frac{D(\rho_1\|\sigma_1)}{D(\rho_2\|\sigma_2)} \right) =
		\frac{\sqrt{V(\rho_1\|\sigma_1)}}{D(\rho_2\|\sigma_2)} Z_{\nu}^{-1}(\eps) \, , \label{eq:conj1}
	\end{align}
	where $Z_{\nu}(\cdot)$ is the cumulative Rayleigh-normal distribution function~\cite{kumagai13}.
\end{conjecture}

This form is of interest since it shows a resonance behavior when $\nu = 1$, where the contribution in the second-order term turns positive even for arbitrarily small~$\eps$. This means that there exist pairs of states that can be transformed into each other without loss due to finite size effects (up to second order). This effect has been observed both analytically and numerically in the commutative case~\cite{korzekwa18b}, and its applications to fully quantum resource theories remain to be explored.

The limit expression in~\eqref{eq:conj1} was shown to hold for the case where $\sigma_1$ and $\sigma_2$ are both proportional to the identity in the work of Kumagai and Hayashi~\cite{kumagai13}, and the above conjecture thus constitutes a natural fully quantum generalization of their result.
Building on that and an embedding technique from quantum thermodynamics, the equality was also shown for general $\sigma_1$ and $\sigma_2$ as long as $\rho_1$ and $\rho_2$ commute with $\sigma_1$ and $\sigma_2$, respectively~\cite{chubb18}.\footnote{To be more precise, the cited work~\cite{chubb18} only shows this for the case where $\sigma_1 = \sigma_2$, but the required extension follows after a close inspection of the proof.} 
The same special case can also be solved in the moderate error regime by adapting the results in~\cite{korzekwa18}.

Finally, note that since we are now concerned with higher order contribution that are a function of the error threshold $\eps$, the limit does in general depend on how exactly we measure the error. It is thus not obvious how an appropriate conjecture for the trace distance would look like, for example.

\paragraph{Acknowledgements.}
FB acknowledges partial support from the Japan Society for the Promotion of Science (JSPS) KAKENHI, Grant No.19H04066, and the program for FRIAS-Nagoya IAR Joint Project Group. DS acknowledges support from the Swiss National Science Foundation via the NCCR QSIT as well as project No.~200020\_165843.






\end{document}